\documentclass[10pt, conference]{IEEEtran}
\usepackage{epsfig,amsfonts,subfigure,color,amsbsy}
\usepackage[nolist]{acronym}
\usepackage{graphicx,cite,amssymb,amsthm,bm}
\usepackage{amsmath}
\usepackage{mathrsfs}
\usepackage{bm}
\usepackage{multirow}
\usepackage{bigstrut}
\usepackage{psfrag}
\usepackage{stackengine}
\usepackage{xcolor}
\usepackage{mathtools}
\usepackage{tabularx,array}
\usepackage{mathbbol}
\usepackage{lipsum}
\usepackage[sans]{dsfont}
\mathtoolsset{showonlyrefs}
\usepackage{array, makecell}
\usepackage{float}
\usepackage{balance}
\usepackage{verbatim}
\usepackage{mathrsfs}
\usepackage{flushend}

\usepackage{tikz}
\usetikzlibrary{shapes,arrows}
\usetikzlibrary{calc}

 \IEEEoverridecommandlockouts\IEEEpubid{\makebox[\columnwidth]{ 978-1-6654-3540-6/22/\$31.00~\copyright~2022 IEEE \hfill} \hspace{\columnsep}\makebox[\columnwidth]{ }}

\theoremstyle{remark}
\theoremstyle{definition}

\newtheorem{rem}{Remark}

\newtheorem{theorem}{Theorem}
\newtheorem{corollary}{Corollary}

\renewcommand{\[}{\begin{equation}}
\renewcommand{\]}{\end{equation}}

\renewcommand{\mathbf}{\bm}

\newcommand{\de}{\mathrm{d}}

\newcommand{\circConv}{\mathop{\vphantom{\sum}\mathchoice
  {\vcenter{\hbox{\huge $\circledast$}}}
  {\vcenter{\hbox{\Large A}}}{\mathrm{A}}{\mathrm{A}}}\displaylimits}
\newcommand{\hadProd}{\mathop{\vphantom{\sum}\mathchoice
  {\vcenter{\hbox{\huge $\odot$}}}
  {\vcenter{\hbox{\Large A}}}{\mathrm{A}}{\mathrm{A}}}\displaylimits}

\DeclareMathOperator{\weight}{wt}
\DeclareMathOperator{\dist}{d}

\newcommand{\ring}{\mathbb{Z}}
\newcommand{\unitset}{\mathbb{Z}_q^{\times}}
\newcommand{\LW}{\weight_{\scriptscriptstyle\mathsf{L}}}
\newcommand{\LD}{\dist_{\scriptscriptstyle\mathsf{L}}}

\newcommand{\code}{\mathcal{C}}
\newcommand{\ens}{\mathscr{C}}
\renewcommand{\H}{\bm{H}}

\newcommand{\vecu}{\bm{u}}
\newcommand{\vecv}{\bm{v}}
\newcommand{\vecx}{\bm{x}}
\newcommand{\vecy}{\bm{y}}

\newcommand{\vece}{\bm{e}}
\newcommand{\vecE}{\bm{E}}
\newcommand{\vecf}{\bm{f}}
\newcommand{\vecL}{\bm{L}}
\newcommand{\vecc}{\bm{\phi}}

\newcommand{\permMat}{\bm{\Pi}}

\newcommand{\entH}{\mathrm{H}_e}

\newcommand{\vn}{\mathsf{v}}
\newcommand{\cn}{\mathsf{c}}

\newcommand{\neigh}[1]{\ensuremath{\mathcal{N}\left(#1\right)}}
\newcommand{\msg}[2]{\ensuremath{m_{#1\rightarrow#2}}}
\newcommand{\mch}{\ensuremath{m_{\mathsf{ch}}}}
\newcommand{\mchVec}{\ensuremath{\bm{m}_{\mathsf{ch}}}}
\newcommand{\msgVec}[2]{\ensuremath{\bm{m}_{#1\rightarrow#2}}}
\newcommand{\msgRV}[2]{\ensuremath{M_{#1\rightarrow#2}^{(\ell)}}}

\newcommand{\thr}{\ensuremath{\delta^\star}}
\newcommand{\thrSMP}{\ensuremath{\delta^\star_{\scriptscriptstyle\mathsf{SMP}}}}
\newcommand{\thrNBP}{\ensuremath{\delta^\star_{\scriptscriptstyle\mathsf{BP}}}}
\newcommand{\thrSH}{\ensuremath{\delta^\star_{\scriptscriptstyle\mathsf{SH}}}}

\begin{document}


\title{Analysis of Low-Density Parity-Check Codes over Finite Integer Rings for the Lee Channel}

\author{
\IEEEauthorblockN{Jessica Bariffi}
\IEEEauthorblockA{
\textit{German Aerospace Center}\\
Wessling, Germany \\
jessica.bariffi@dlr.de}
\and
\IEEEauthorblockN{Hannes Bartz}
\IEEEauthorblockA{
\textit{German Aerospace Center}\\
Wessling, Germany \\
hannes.bartz@dlr.de}
\and
\IEEEauthorblockN{Gianluigi Liva}
\IEEEauthorblockA{
\textit{German Aerospace Center}\\
Wessling, Germany \\
gianluigi.liva@dlr.de}
\and
\IEEEauthorblockN{Joachim Rosenthal}
\IEEEauthorblockA{
\textit{University of Zurich}\\
Zurich, Switzerland \\
rosenthal@math.uzh.ch}
    \thanks{J. Rosenthal has been supported
    in part by the Swiss National Science Foundation under the grant No. 188430. J. Bariffi, H. Bartz and G. Liva acknowledge the financial support by the Federal Ministry of Education and Research of Germany in the programme of "Souver\"an. Digital. Vernetzt." Joint project 6G-RIC, project identification number: 16KISK022.
    }}

 \maketitle

 \IEEEoverridecommandlockouts
 



\begin{abstract}
We study the performance of nonbinary low-density parity-check (LDPC) codes over finite integer rings over two channels that arise from the Lee metric. The first channel is a discrete memory-less channel (DMC) matched to the Lee metric. The second channel adds to each codeword an error vector of constant Lee weight, where the error vector is picked uniformly at random from the set of vectors of constant Lee weight. It is shown that the marginal conditional distributions of the two channels coincide, in the limit of large block length. Random coding union bounds on the block error probability are derived for both channels. Moreover, the performance of selected LDPC code ensembles is analyzed by means of density evolution and finite-length simulations, with belief propagation decoding and with a low-complexity symbol message passing algorithm and it is compared to the derived bounds.
\end{abstract}





\begin{acronym}
    \acro{BP}{belief propagation}
    \acro{VN}{variable node}
    \acro{CN}{check node}
    \acro{DE}{density evolution}
    \acro{EXIT}{extrinsic information transfer}
    \acro{i.i.d.}{independent and identically distributed}
    \acro{LDPC}{low-density parity-check}
    \acro{MAP}{maximum a posteriori probability}
    \acro{MCM}{Monte Carlo method}
    \acro{r.v.}{random variable}
    \acro{p.m.f.}{probability mass function} 
    \acro{ML}{maximum likelihood}
    \acro{WEF}{weight enumerating function}
    \acro{MDS}{maximum distance separable}
    \acro{AWE}{average weight enumerator}
    \acro{DMC}{discrete memory-less channel}
    \acro{w.r.t.}{with respect to}
    \acro{BSC}{binary symmetric channel}
    \acro{SMP}{symbol message-passing}
    \acro{RV}{random variable}
    \acro{PMF}{probability mass function}
    \acro{qSC}[$q$-SC]{$q$-ary symmetric channel}
    \acro{RCU}{random coding union}
    \acro{BMP}{binary message-passing}
    \acro{PEG}{progressive edge growth}
    \acro{LSF}{Lee Symbol Flipping}
    \acro{TV}{total variation}
\end{acronym}


\section{Introduction}\label{sec:not}

The construction of channel codes for the Lee metric \cite{Ulrich,lee1958some} attracted some attention in the past \cite{prange1959use,berlekamp1966negacyclic,golomb1968algebraic,chiang1971channels,etzion2010dense}. Currently, codes for the Lee metric are considered for cryptographic applications \cite{weger2020hardness,weger2020information} thanks to their potential in decreasing the public key size in code-based public-key cryptosystems. Furthermore, codes for the Lee metric have potential applications in the context of magnetic  \cite{roth1994lee} and DNA \cite{gabrys2017asymmetric} storage systems. 

In this paper, we analyze the performance of certain code classes in the context of Lee metric decoding. In particular, we consider two channel models. The first model is a \ac{DMC} \emph{matched} to the Lee metric \cite{massey1967notes,chiang1971channels}, i.e., the \ac{DMC} whose \ac{ML} decoding rule reduces to finding the codeword at minimum Lee distance from the channel output. The second model is a channel that adds to each codeword an error vector of constant Lee weight, where the error vector is picked uniformly at random from the set of length-$n$ vectors of constant Lee weight (here, $n$ is the block length). The first model will be referred to as the \emph{Lee channel}, whereas the second model will be dubbed \emph{constant-weight Lee channel}. It will be shown that the marginal conditional distribution of the constant-weight Lee channel reduces to the conditional distribution of a suitably-defined (memory-less) Lee channel, as $n$ grows large. Random coding bounds are derived for both channels, providing a finite-length performance benchmark to evaluate the block error probability of practical coding schemes. We then study the performance of nonbinary \ac{LDPC} codes \cite{studio3:GallagerBook} over finite rings \cite{Fuja05:Ring}, in the context of Lee-metric decoding. The codes will be analyzed, from a code ensemble viewpoint, via density evolution. Two decoding algorithms will be considered, namely the well-known (nonbinary) \ac{BP} algorithm \cite{DM98,Fuja05:Ring} and the recently-introduced low-complexity \ac{SMP} algorithm \cite{Lazaro19:SMP}, where the latter will be adapted to the Lee channel (the \ac{SMP} was originally defined for $q$ary symmetric channels only). We will compare the performance of the two decoding algorithms to the \ac{LSF} presented in \cite[Algorithm 2]{santini2020low} for LDPC Codes in the Lee metric. The \ac{SMP} decoding algorithm, thanks to its low complexity, is of practical interest for code-based cryptosystems \cite{santini2020low}. To simplify the exposition, the analysis will be limited to regular \ac{LDPC} code ensembles (that are mainly considered for code-based cryptography). The extension of the analysis to irregular and protograph-based \ac{LDPC} code ensembles is straightforward. Finite-length simulation results will be provided for both the Lee and the constant-weight Lee channels, and will be compared with finite-length benchmarks.

The paper is organized as follows. Section \ref{sec:prel} provides some definitions and useful results. The channel models are introduced in Section \ref{sec:leeChannel}, together with finite-length performance bounds.  In Section \ref{sec:lee_ldpc} we analyse the performance of \ac{LDPC} codes over Lee channels. Conclusions follow in Section \ref{sec:conc}.


\section{Preliminaries}\label{sec:prel}

Let $\ring_q$ be the ring of integers modulo $q$. In the following, all logarithms are in the natural base. The set of units of a ring $\ring_q$ is indicated by $\unitset$. We  denote random variables by uppercase letters, and their realizations with lower case letters. 
Moreover, we use the shorthand $[x]^+$ to denote $\max(0,x)$.


The Lee weight \cite{lee1958some} of a scalar $a \in \ring_q$ is
\begin{align*}
	\LW(a) := \min(a, q-a).
\end{align*}
The Lee weight of a vector $\vecx \in \ring_q^n$ is defined to be the sum of the Lee weights of its elements, i.e.,
\[
	\LW(\vecx) = \sum_{i=1}^n \LW(x_i).
\]
Note that the Lee weight of an element $a \in \ring_q$ is upper bounded by $\lfloor q/2 \rfloor$. Hence, the Lee weight of a length-$n$ vector $\vecx$ over $\ring_q$ is upper bounded by $n\cdot \lfloor q/2 \rfloor$. To simplify the notation, we will always denote $r := \lfloor q/2 \rfloor$. We have that
\begin{align}\label{property:symm_leeweight}
	\LW (a) = \LW (q- a) \; \text{ for every } a \in \lbrace 1, \dots , r \rbrace .
\end{align}
The Lee distance of two scalars $a,b \in \ring_q$ is
$\LD (a, b) := \LW(a-b)$.
The Lee distance between  $\vecx, \vecy \in \ring_q^n$ is
\begin{align*}
	\LD (\vecx, \vecy) = \sum_{i=1}^n \LD(x_i,y_i).
\end{align*}

\subsection{Low-Density Parity-Check Codes over Finite Integer Rings}
We will consider $(n,k)$ linear block codes over $\ring_q$ and we denote by $R=k/n$ the code rate. An $(n,k)$ \ac{LDPC} code over $\ring_q$ \cite{Fuja05:Ring} is defined by a $m\times n$ \emph{sparse} matrix $\H$, which can be described via a bipartite graph $\mathcal{G}$ consisting of a set of $n$ \acp{VN} $\{\vn_0,\vn_1,\ldots,\vn_{n-1}\}$ and a set of $m$ \acp{CN} $\{\cn_0,\cn_1,\ldots,\cn_{m-1}\}$ where the \ac{VN} $\vn_j$ is connected with an edge to the \ac{CN} $\cn_i$ if and only if the entry $h_{i,j}$ in $\H$ is nonzero. The degree of a node refers to the number of edges that are connected to the node. The neighbors of a \ac{VN} $\vn$ is the set $\neigh{\vn}$ composed by \acp{CN} that are connected to $\vn$ by an edge. Similarly, the neighbors of a \ac{CN} $\cn$ is the set $\neigh{\cn}$ composed by \acp{VN} that are connected to $\cn$ by an edge. We denote by $\ens_{v,c}^n$ the unstructured regular (length-$n$) \ac{LDPC} code ensemble, i.e., the set of codes defined by an $m \times n$ matrix $\H$ whose bipartite graph possesses constant \ac{VN} degree $v$ and constant \ac{CN} degree $c$.
We denote the ensemble design rate as $R_0=1-m/n$. 
When sampling an \ac{LDPC} code from the given ensemble, we assume the nonzero entries drawn independently and uniformly from $\unitset$ as proposed in \cite{Fuja05:Ring}.

\subsection{Useful Results and Definitions}

Letting $a_n$ and $b_n \neq 0$ be two real-valued sequences, we say that $a_n$ and $b_n$ are exponentially equivalent as $n \rightarrow \infty$, writing $a_n \doteq b_n$ if and only if \cite[Ch. 3.3]{CoverThomasBook}
\begin{align*}
	\lim_{n \rightarrow \infty} \frac{1}{n}\log \left(\frac{a_n}{b_n}\right) =0 \, .
\end{align*}
We denote by $\vecf(\vecx)=(f_0(\vecx),f_1(\vecx),\ldots,f_{q-1}(\vecx))$ the composition (i.e., empirical distribution) of a vector $\vecx\in\ring_q^n$, i.e., $f_i$ is the relative frequency of $i$ in $\vecx$. We introduce the set
\begin{equation}
	\mathcal{S}_{n\delta}^n:=\left\{\vecx\, \big| \, \vecx\in \ring^n_q, \LW(\vecx)=n\delta\right\}.
\end{equation}
Here, $\mathcal{S}_{n\delta}^n$ defines the surface of radius-$n\delta$ $n$-dimensional Lee sphere. The set of vectors in $\ring^n_q$ with composition $\vecc$ is 
\begin{equation}
	\mathcal{T}_{\vecc}^n:=\left\{\vecx\, \big| \, \vecx\in \ring^n_q, \vecf(\vecx)=\vecc\right\}.
\end{equation}
We have that \cite[Ch. 11.1]{CoverThomasBook}
\begin{equation}
	\left|\mathcal{T}_{\vecc}^n\right| \doteq \exp(n\entH(\vecc)) \label{eq:exponential_equiv}
\end{equation}
where 
\begin{equation}
\entH(\vecc):=-\!\!\sum_{i=0, \phi_i\neq 0}^{q-1}\!\! \phi_i \log \phi_i.\label{eq:entropy}
\end{equation}


\section{The Lee Channel}\label{sec:leeChannel}
For $x, y, e \in \ring_q$ consider the \ac{DMC} 
\begin{equation}
	y=x+e
\end{equation} 
where $y$ is the channel output, $x$ the channel input, and $e$ is an additive error term. More specifically, we restrict to the case where $e$ is a realization of a \ac{RV} $E$, distributed as $P_E(e)\propto \exp(-\beta \LW(e))$,
where $\beta > 0$ is a constant that defines (together with the alphabet) the channel. Defining the normalization constant 
\[
	Z(\beta):= \sum_{e=0}^{q-1}\exp(-\beta \LW(e))   \label{eq:partition_function} 
\] we get the channel law
\begin{equation}
	P_{Y|X}(y|x)=\frac{1}{Z}\exp\left(-\beta\LD(x,y)\right).\label{eq:Lee_channel}
\end{equation}
We refer next to the channel defined in \eqref{eq:Lee_channel} as the \emph{Lee channel}. We denote the expectation of $\LW(E)$ as $\delta$, given by \cite{mezard2009information}
\[
\delta = -\frac{\de \log Z(\beta)}{\de \beta}. \label{eq:delta}
\]
Our interest in  \eqref{eq:Lee_channel} stems from two observations:
\begin{itemize}
\item[i.] The channel defined in \eqref{eq:Lee_channel} is the \ac{DMC} matched to the Lee metric \cite{massey1967notes,chiang1971channels}, i.e., the channel whose \ac{ML} decoding rule reduces to finding the codeword $\vecx\in\code$ that minimizes the Lee distance from the channel output $\vecy$;
\item[ii.] The conditional distribution \eqref{eq:Lee_channel} arises (in the limit of large $n$) as the marginal distribution of a channel (in the following, referred to as an \emph{constant-weight Lee channel}) adding to the transmitted codeword an error pattern drawn uniformly at random from a set of vectors of constant Lee weight. This is especially interesting for code-based public-key cryptosystems in the Lee metric \cite{weger2020hardness}.
\end{itemize}
A derivation of the result in ii. is given next.

\subsection{Marginal Distribution of Constant-Weight Lee Channels}\label{subsec:marg_distrib}
Consider a constant-weight Lee channel 
\begin{equation}
	\vecy=\vecx+\vece
\end{equation} 
with $\vecy, \vecx, \vece \in \mathbb{Z}^n_q$, and where $\vece$ is drawn, with uniform probability, from the set $\mathcal{S}_{n\delta}^n$.
We have 
$P_{\vecE}(\vece)=\left|\mathcal{S}_{n\delta}^n\right|^{-1}$ for all $\vece\in\mathcal{S}_{n\delta}^n$, with $P_{\vecE}(\vece)=0$ otherwise. We are interested the marginal distribution $P_E(e)$ in the limit for $n\rightarrow \infty$. The marginal distribution plays an important role, for instance, in the initialization of iterative decoders of \ac{LDPC} codes, when used over constant-weight Lee channels \cite{santini2020low}. While the focus here is in the asymptotic (in the block length $n$) case, the derived marginal distribution provides an excellent approximation of the true marginal down to moderate-length blocks ($n$ in the order of a few hundreds).
The derivation follows by seeking the composition that dominates the set $\mathcal{S}_{n\delta}^n$. More specifically, we should look for the empirical distribution $\vecc$ that maximizes the cardinality of $\mathcal{T}_{\vecc}^n$ under the constraint 
\[
\sum_{i=0}^{q-1} \LW(i) \phi_i=\delta.\label{eq:constraint}
\]
The task is closely related to the problem, in statistical mechanics, of finding the distribution of a systems state by relating it to that states energy and temperature  \cite{boltzmann1868studien,gibbs1902elementary}. Owing to \eqref{eq:exponential_equiv}, and taking the limit for $n\rightarrow \infty$, we will look for the empirical distribution maximizing the entropy \eqref{eq:entropy} under the constraint \eqref{eq:constraint} \cite[Ch. 12]{CoverThomasBook}, i.e.
\[
\vecc^\star = \arg\max \entH(\vecc) 
\]
with 
$
\sum_{i=0}^{q-1} \LW(i) \phi_i=\delta$.
By introducing the Lagrange multiplier $\beta$, we aim at finding the maximum in $\vecc$ of 
\[
\mathsf{f}(\vecc,\beta):=\entH(\vecc) - \beta\left(\sum_{i=0}^{q-1} \LW(i) \phi_i-\delta\right).
\]
The result yields the distribution
\[
\phi_i^\star=\frac{1}{Z}\exp\left(-\beta\LW(i)\right) \label{eq:boltzmann} 
\]
with $Z$ given in \eqref{eq:partition_function} and $\beta$ obtained by enforcing the condition \eqref{eq:constraint} (i.e., by solving \eqref{eq:delta} in $\beta$). The distribution \eqref{eq:boltzmann} is closely related to the Boltzmann distribution \cite{boltzmann1868studien,CoverThomasBook}, which may be recovered by interpreting the Lee weight $\LW(i)$ as an energy value. Notably, when drawing $\vece$ with uniform probability from the set $\mathcal{S}_{n\delta}^n$, $\vece$ will possess an empirical distribution close to $\vecc^\star$ with high probability as $n$ grows large. The result follows by the conditional limit theorem \cite[Theorem 11.6.2]{CoverThomasBook}.

\newcommand{\Hdelta}{\mathsf{H}_\delta}
\newcommand{\Hdeltaplus}{\mathsf{H}^+_\delta}
\subsection{Bounds on the Block Error Probability}
Denote the natural entropy of a random variable distributed according to $\vecc^\star$ with mean $\delta$, $\entH(\vecc^\star)$, as $\Hdelta$. Moreover, let
\[
\delta_q := \left\{
\begin{array}{ll}
     (q^2-1)/4q&  \text{if }q\text{ is odd}\\
     q/4& \text{if }q\text{ is even}
\end{array}
\right.
\]
and the function
\[
\Hdeltaplus := \left\{
\begin{array}{ll}
     \Hdelta&  \text{if }\delta\leq \delta_q\\
     \log q& \text{otherwise}.
\end{array}
\right.
\]
The following theorem establishes a \ac{RCU} bound, providing an upper bound on the error probability, $P_B(\code)$, achievable by the best $(n,nR)$ code $\code$ on $\ring_q$ over a constant-weight Lee channel with normalized  error vector weight equal to $\delta$.
\begin{theorem}\label{th:RCUconst}
The expected error probability of a random $(n,nR)$ code $\code$ on $\ring_q$ when used to communicate over a constant-weight Lee channel with normalized weight of the error vector equal to $\delta$ satisfies
\begin{equation}
    \mathbb{E}\left[P_B(\code)\right]<\exp\left(-n\left[(1-R)\log q - \Hdeltaplus\right]^+\right).
\end{equation}
\end{theorem}
\begin{proof}
    The proof of Theorem \ref{th:RCUconst} is based on \cite[Theorem 16]{CCR10}, where for the evaluation of the pair-wise error probability we determine the probability of generating a random codeword that lies within a Lee sphere of radius $n\delta$. By noticing that the volume of such a sphere is tightly upper bounded by $\exp(n \Hdeltaplus)$, the result follows.
\end{proof}
Note that the bound provided in Theorem \ref{th:RCUconst} can easily be extended to the memoryless Lee channel by averaging over the distribution of the Lee weight $D$ of the error pattern, yielding the following corollary.
\begin{corollary}\label{th:RCU}
The expected error probability of a random $(n,nR)$ code $\code$ on $\ring_q$ when used to communicate over a memoryless Lee channel with parameter $\delta$ satisfies
\begin{equation}
    \mathbb{E}\left[P_B(\code)\right]<\mathbb{E}\left[\exp\left(-n\left[(1-R)\log q - {\mathsf{H}^+_{D/n}}\right]^+\right)\right].
\end{equation}
\end{corollary}

\section{LDPC Codes: Analysis over the Lee Channel}\label{sec:lee_ldpc}

We review first two message-passing decoders for nonbinary \ac{LDPC} codes, i.e., the well-known \ac{BP} algorithm \cite{DM98,Fuja05:Ring} and the \ac{SMP} algorithm introduced in \cite{Lazaro19:SMP}. We then analyze the performance achievable by the two algorithms in an asymptotic setting (via \ac{DE} analysis) and at finite block length (via Monte Carlo simulations).

\subsection{Message-Passing Decoders}

\subsubsection{Belief Propagation Decoding}
We now consider first \ac{BP} decoding of nonbinary \ac{LDPC} codes defined on rings. The decoding algorithm is outlined below.

\begin{enumerate}
    \item \textbf{Initialization.} Define the likelihood at \ac{VN} by $\vn$ $\mchVec:=\left(P_{Y|X}(y\mid 0),\dots,P_{Y|X}(y\mid q-1)\right)$, i.e., $\mchVec$ is the \ac{PMF} associated with the channel observation for the \ac{VN} $\vn$. Let $\permMat_{\cn,\vn}$ be the permutation matrix induced by the parity-check matrix element $h_{\cn,\vn}$ (associated with the edge between \ac{CN} $\cn$ and \ac{VN} $\vn$). In the first iteration, each \ac{VN} $\vn$ sends to all $\cn\in\neigh{\vn}$ the message 
     \begin{equation}
         \msgVec{\vn}{\cn}=\mchVec\permMat_{\cn,\vn}
     \end{equation}
    \item \textbf{CN-to-VN step.} Let $\circledast$ denote the circular convolution of the \acp{PMF}. Then, each \ac{CN} $\cn$ computes \label{itm:CN-to-VN-NBP}
    \begin{equation}
\vecu=\circConv_{\vn'\in\neigh{\cn}\setminus\{\vn\}}\msgVec{\vn'}{\cn}
    \end{equation}
    Given the $(q\times q)$ inverse permutation matrix, $\permMat_{\vn,\cn}^{-1}$, associated with $h_{\cn,\vn}^{-1}$ the \ac{CN}-to-\ac{VN} message is then
    \begin{equation}
        \msgVec{\cn}{\vn}=\vecu\cdot\permMat_{\cn,\vn}^{-1}.
    \end{equation}
    \item \textbf{VN-to-CN step.} We denote by $\odot$ the element-wise Hadamard (or Schur) product of the \acp{PMF}, and by $K$ a normalization constant enforcing $\sum_{i=0}^{q-1}v_i=1$. Each \ac{VN} $\vn$ then computes \label{itm:VN-to-CN-NBP}
    \begin{equation}
\vecv=K\!\!\!\!\hadProd_{\cn'\in\neigh{\vn}\setminus\{\cn\}}\!\!\!\!\msgVec{\cn'}{\vn}
    \end{equation}
    and then sends to its neighboring \ac{CN} $\cn$ the message 
    \begin{equation}
        \msgVec{\vn}{\cn}=\vecv\cdot\permMat_{\cn,\vn}. 
    \end{equation}

    \item \textbf{Final decision.} After iterating steps~\ref{itm:CN-to-VN-NBP} and~\ref{itm:VN-to-CN-NBP} at most $\ell_\text{max}$ times, the final decision at each \ac{VN} $\vn$ is 

    \begin{equation}
        \hat{x}=\underset{x\in\ring_q}{\arg\max}\, v_x^{\scriptscriptstyle \textsf{APP}}
    \end{equation}
    where 
    \begin{equation}
       \vecv^{\scriptscriptstyle \textsf{APP}}=\hadProd_{\cn\in\neigh{\vn}}\msgVec{\cn}{\vn}.
    \end{equation}
\end{enumerate}

\subsubsection{Symbol Message Passing (SMP) Decoding}
Under \ac{SMP} decoding each message exchanged by a \ac{VN}/\ac{CN} pair is a symbol, i.e., an hard estimate of the codeword symbol associated with the \ac{VN}. Thanks to this, \ac{SMP} allows remarkable savings in the internal decoder data flow, compared to \ac{BP} decoding. Following the principle outlined in \cite{lechner2011analysis}, the messages from \acp{CN} to \acp{VN} are modeled as observations at the output a $q$-ary input, $q$-ary output \ac{DMC}. By doing so, the messages at the input of each \ac{VN} can be combined by multiplying the respective likelihoods (or by summing the respective log-likelihoods).

Given a \ac{DMC} $P_{Y| X}(y|x)$ and a channel output $y\in\ring_q$, we define the log-likelihood vector ($\vecL$-vector)
\begin{equation}
    \vecL(y):=(L_0(y),L_1(y),\dots,L_{q-1}(y)))
\end{equation}
where $L_x(y)=\log\left(P_{Y| X}(y\mid x)\right)$.

\begin{enumerate}
    \item \textbf{Initialization.} 
    Each \ac{VN} $\vn$ sends the corresponding Lee channel observation
        $\msg{\vn}{\cn}=y$
    to all $\cn\in\neigh{\vn}$.

    \item \textbf{CN-to-VN step.} Each \ac{CN} $\cn$ computes \label{itm:CN-to-VN-SMP}
    \begin{equation}
        \msg{\cn}{\vn} =h_{\cn,\vn}^{-1}\sum_{\vn'\in\neigh{\cn}\setminus\{\vn\}} h_{\cn,\vn'}\msg{\vn'}{\cn}.
    \end{equation}
    
    \item \textbf{VN-to-CN step.} Define the aggregated extrinsic $\vecL$-vector \label{itm:VN-to_CN-SMP}
    \begin{equation}\label{eq:def_E_vec}
\vecE=\vecL(y)+\sum_{\cn'\in\neigh{\vn}\setminus\{\cn\}}\vecL\left(\msg{\cn'}{\vn}\right).
    \end{equation}
    Under the \ac{qSC} approximation, each extrinsic channel from \ac{CN} $\cn'$ to \ac{VN} $\vn$ is modeled according to 
    \begin{equation}
        P_{M| X}(m| x)=
        \begin{cases}
            1-\xi & \text{if } m=x
            \\ 
            \xi/(q-1) & \text{otherwise}
        \end{cases}\label{eq:QSCapproxdec}
    \end{equation}
    where, for the sake of computing $\vecL(\msg{\cn'}{\vn})$, the iteration-dependent extrinsic channel error probability $\xi$ can be obtained from the \ac{DE} analysis (as described in the Section \ref{subsubsec:DESMP}).
    Then the \ac{VN}-to-\ac{CN} messages are 
    \begin{equation}
        \msg{\vn}{\cn}=\underset{x\in\ring_q}{\arg\max} \, E_x.
    \end{equation}
    
    \item \textbf{Final decision.} After iterating steps \ref{itm:CN-to-VN-SMP} and \ref{itm:VN-to_CN-SMP} at most $\ell_\text{max}$ times, the final decision at each \ac{VN} $\vn$ is 
    \begin{equation}
        \hat{x}=\underset{x\in\ring_q}{\arg\max}\, L_x^{\scriptscriptstyle \textsf{FIN}}
    \end{equation}
    where
    
    \begin{equation}\label{eq:final}
        \vecL^{\scriptscriptstyle \textsf{FIN}}=\vecL(\mch)+\sum_{\cn\in\neigh{\vn}}\vecL\left(\msg{\cn}{\vn}\right).
    \end{equation}
\end{enumerate}

\subsection{The $q$SC-Assumption}
The choice of the \ac{DMC} used to model the extrinsic channel plays a crucial role for the performance of the \ac{SMP} algorithm. In \cite{lechner2011analysis}, for the case of \ac{BMP} decoding, it was suggested to model the \ac{VN} inbound messages as observations of a \ac{BSC}, whose transition probability was estimated by  \ac{DE} analysis. The approach was generalized in \cite{Lazaro19:SMP} for \ac{SMP}, where the \ac{VN} inbound messages are modelled as observations of a \ac{qSC}. We will also model the extrinsic channel as a \ac{qSC} defined in \eqref{eq:QSCapproxdec}, although in our setting the model holds only in an approximate sense.
The use of the \ac{qSC} approximation is particularly useful from a practical viewpoint since it simplifies the \ac{VN} processing in \ac{SMP} decoding. Note moreover that for \ac{LDPC} codes over finite fields, the extrinsic channel transition probabilities, averaged over a uniform distribution of nonzero elements in the parity-check matrix, yield a \ac{qSC} \cite{Lazaro19:SMP}.

For the case of $\ring_q$ where $q$ is non-prime, the average extrinsic channel transition probabilities do not describe a \ac{qSC}. Nevertheless, if the units $x\in\unitset$ are used to label the graph edges with uniform probability, we expect that for an integer ring consisting of relatively many units the \ac{qSC} approximation should turn to be accurate. To provide some empirical evidence of this conjecture, we adopt the methodology used in \cite{xie2006accuracy} to support the use of the Gaussian approximation in the \ac{DE} analysis of \ac{BP} decoding of binary \ac{LDPC} codes. In particular, we show numerically that the \ac{TV} distance between the extrinsic channel distribution and the \ac{qSC} is generally small, and it vanishes with the number of iterations. The \ac{TV} distance of two probability distributions $P$ and $Q$ over the same discrete alphabet $\mathcal{X}$ is defined as \cite[Proposition 4.2]{wilmer2009markov} 
\begin{align*}
    \mathsf{TV}(P, Q) := \dfrac{1}{2} \sum_{x \in \mathcal{X}} \left\vert P(x) -Q(x) \right\vert .
\end{align*}
In Figure \ref{fig:TV_(3,6)} and \ref{fig:TV_(4,8)} we show for $\ring_8$, $\ring_9$ and $\ring_{12}$, that the \ac{TV} indeed tends to zero when performing Monte Carlo simulations for different numbers of iterations of the \ac{SMP} decoder for some choices of $\delta$ and different regular LDPC code ensembles. 
Iterative decoding thresholds of some regular nonbinary \ac{LDPC} code ensembles can be found in Table \ref{tab:dec_thresholds}. We have chosen these three finite integer rings to cover different cases for the relative number of unit elements, i.e. the fraction of units $\mathcal{U}_q$ in $\ring_q$ are  $\mathcal{U}_8 = 1/2$, $\mathcal{U}_9 = 2/3$ and $\mathcal{U}_{12} = 1/3$. 
 The figures support the statement that, for integer rings with relatively few unit elements, the approximation  is less accurate in the first iterations. Note that the first few iterations play an important role in determining the iterative decoding threshold.

\begin{figure}[htp]
    \centering
    \includegraphics[width = \linewidth]{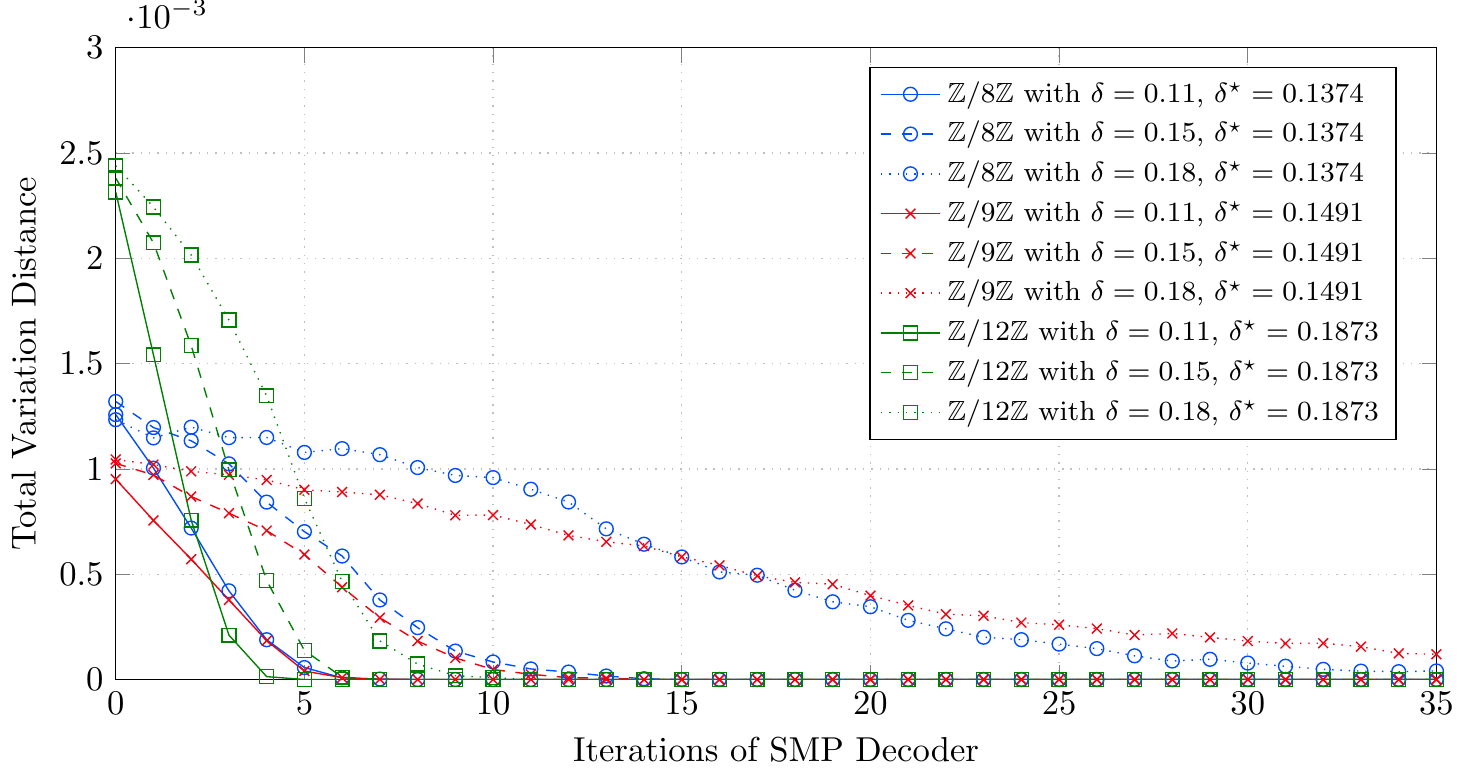}

    \caption{Evolution of the \ac{TV} distance between the extrinsic channel distribution and the \ac{qSC} for regular $(3, 6)$ LDPC code ensembles in the \ac{SMP} decoder.}
    \label{fig:TV_(3,6)}
\end{figure}
\begin{figure}[htp]
    \includegraphics[width = \linewidth]{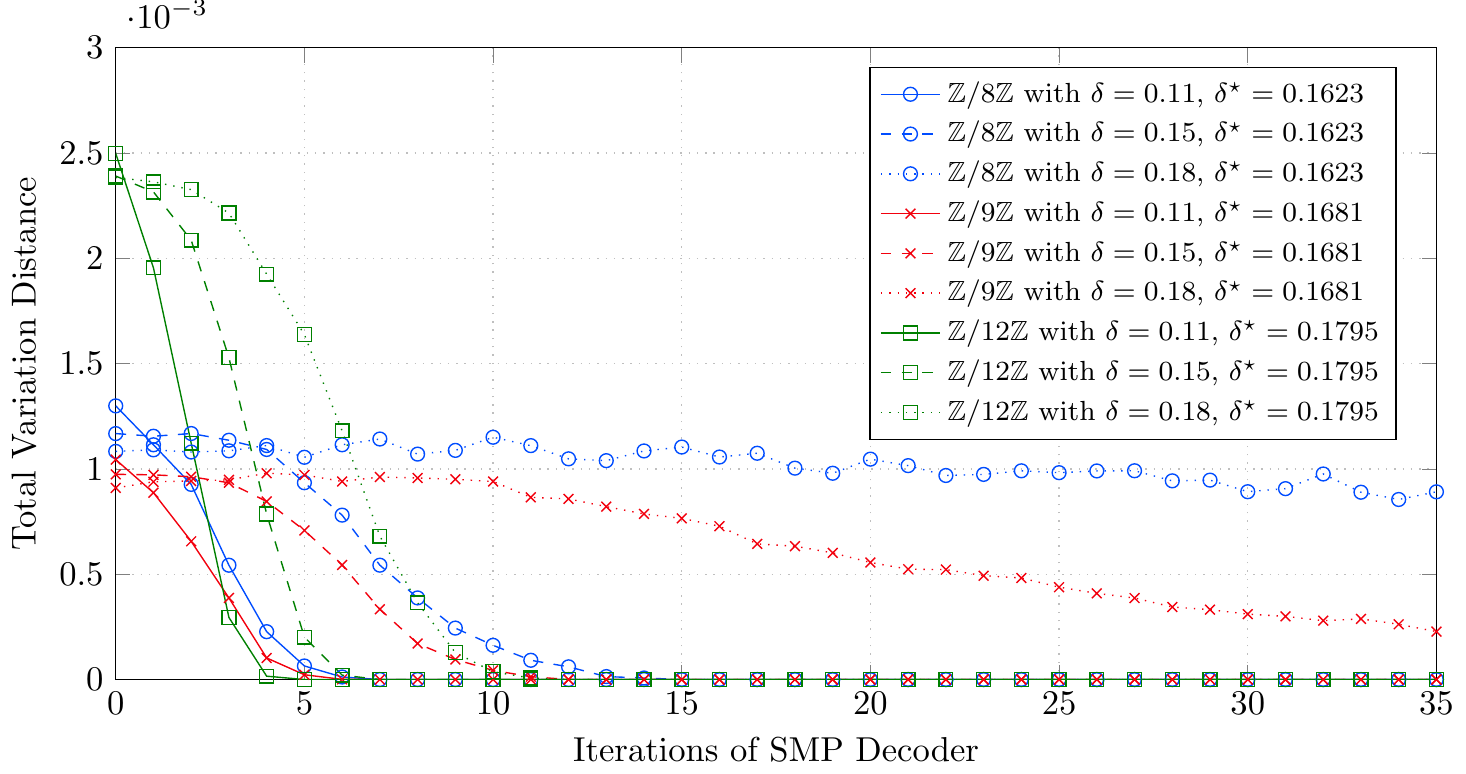}

    \caption{Evolution of the \ac{TV} distance between the extrinsic channel distribution and the \ac{qSC} for regular $(4, 8)$ LDPC code ensembles in the \ac{SMP} decoder.}
    \label{fig:TV_(4,8)}
\end{figure}

\subsection{Density Evolution Analysis}

We analyze next the performance of regular \ac{LDPC} code ensembles on $\ring_q$, over the Lee channel, from a \ac{DE} viewpoint. In particular, we estimate the iterative decoding threshold over the Lee channel~\eqref{eq:Lee_channel} under \ac{BP} and \ac{SMP} decoding.
The iterative decoding threshold $\thr$ is the largest value of the channel parameter $\delta$ \eqref{eq:delta} for which, in the limit of large $n$ and large $\ell_\text{max}$, the symbol error probability of code picked randomly from the ensemble becomes vanishing small~\cite{studio3:richardson01capacity}.
For \ac{BP} decoding, we resort to the \ac{MCM}, while for \ac{SMP} decoding the analysis is outlined next.

\subsubsection{Density Evolution Analysis for \ac{SMP}}\label{subsubsec:DESMP}

The \ac{DE} analysis for \ac{SMP} plays a two-fold role: it allows to estimate the decoding threshold $\thrSMP$ and it provides estimates for the error probabilities $\xi$ of the extrinsic \ac{qSC} which have to be used by the decoder in \eqref{eq:def_E_vec}, \eqref{eq:final}.
We now briefly sketch the \ac{DE} analysis for \ac{SMP} over a \ac{qSC} from~\cite[Sec.~IV]{Lazaro19:SMP} and highlight the respective modifications to estimate the iterative decoding threshold $\thrSMP$ as well as the extrinsic channel error probabilities $\xi$ for transmissions over the Lee channel~\eqref{eq:Lee_channel}.

Due to the linearity of the code and the symmetry of the Lee channel, for the analysis we assume the transmission of the all-zero codeword. Let $\msgRV{\vn}{\cn}$ denote the messages from \ac{VN} $\vn$ to \ac{CN} $\cn$ in the $\ell$-th iteration and define 
\begin{equation}
    p_a^{(\ell)}:=\Pr\left\{\msgRV{\vn}{\cn}=a\mid X=0\right\}.
\end{equation}
For the Lee channel we initialize the \ac{DE} routine from~\cite[Sec.~IV]{Lazaro19:SMP} with the probabilities 
$p_a^{(0)}=P_{Y|X}(a\mid 0)$, $\forall a\in\ring_q$,
where $P_{Y|X}(y|x)$ is the Lee channel transition probability from~\eqref{eq:Lee_channel}.
The remaining steps of the \ac{DE} analysis remain the same as in~\cite[Sec.~IV]{Lazaro19:SMP} except for the definition of the aggregated extrinsic $\vecL$-vector $\bm{E}$ in~\eqref{eq:def_E_vec}.
For the Lee channel, the entries of $\bm{E}$ in the $\ell$-th iteration are given by 
\begin{equation}
    E_b^{(\ell)}=L_0(b)+\mathsf{D}(\xi^{(\ell)})f_b^{(\ell-1)} \qquad\forall b\in\ring_q
\end{equation}
where $\mathsf{D}(\epsilon):=\log(1-\epsilon)-\log(\epsilon/(q-1))$, $\xi^{(\ell)}$ denotes the extrinsic channel error probability and $f_b^{(\ell)}$ denotes the number of \ac{CN}-to-\ac{VN} message taking the value $b\in\ring_q$ in the $\ell$-th iteration.
The decoding threshold is then obtained as the maximum expected normalized Lee weight $\thrSMP$ of a Lee channel distribution~\eqref{eq:Lee_channel} such that $p_{0}^{(\ell)}\rightarrow 1$ as $\ell\rightarrow\infty$.
Decoding thresholds for $\ens_{3,6}$ and  $\ens_{4,8}$ regular \ac{LDPC} code ensembles with $q$ ranging from $5$ to $8$ are given in Table~\ref{tab:dec_thresholds}, as well as the Shannon limit $\thrSH$ for rate $R = 1/2$.

\renewcommand{\arraystretch}{1}
\begin{table}
\caption{Decoding thresholds for regular nonbinary \ac{LDPC} code ensembles under \ac{BP} and \ac{SMP} decoding.}
\label{tab:dec_thresholds}

\centering
\renewcommand{\arraystretch}{1.2}
\begin{tabular}{c|c|c|c|c}

 \hline
 $~~~q~~~$ & $~~~(v,c)~~~$ & $~~~\thrNBP~~~$ & $~~~\thrSMP~~~$ & $~~~\thrSH~~~$
 \\ \hline\hline 
 \multirow{2}{*}{$5$} & $(3,6)$ & $0.2148$ &   $0.1039$ & \multirow{2}{*}{0.2684}
 \\
  & $(4,8)$ & $0.1802$ & $0.1200$ &
 \\\hline
 \multirow{2}{*}{$7$} & $(3,6)$ & $0.3086$ &  $0.1261$ & \multirow{2}{*}{0.3560}
 \\
  & $(4,8)$ & $0.2686$ & $0.1539$  &
 \\\hline
 \multirow{2}{*}{$8$} & $(3,6)$ & $0.3135$ & $0.1374$ & \multirow{2}{*}{0.3950}
 \\
  & $(4,8)$ & $0.26904$ & $0.1623$ &
 \\\hline
\end{tabular}
\end{table}

\subsection{Numerical Results}
In the following, we present numerical results for both \ac{BP} and \ac{SMP} decoding and we compare them to the \ac{LSF} decoder, for which we assumed a decoding threshold $\tau = \frac{d_\vn}{2}$, where $d_\vn$ denotes the variable nodes degree, as the authors suggest. The results, provided in terms of block error rates for $(3,6)$ regular nonbinary \ac{LDPC} codes of length $256$ symbols, are obtained via Monte Carlo simulations. The codes parity-check matrices have been designed via the \ac{PEG} algorithm \cite{HEA05}, with the nonzero coefficients drawn independently and uniformly in $\unitset$.
For the constant-weight Lee channel, the error vectors are drawn uniformly at random from the set of vectors with a given weight. For the case of the (memoryless) Lee channel, we computed a finite-length performance benchmark via the normal approximation  of \cite{CCR10}.

Figure \ref{fig:Lee256}, shows the block error probability over memoryless Lee channels. The impact of the order $q$ on the achievable performance is well captured by the \ac{RCU} bounds. In particular, for a given target block error rate, a larger average normalized Lee weight $\delta$ can be supported for larger $q$. The result applies to the performance of the $(3,6)$ \ac{LDPC} codes as well, under both \ac{BP} and \ac{SMP} decoding, with one key exception: while under \ac{BP} decoding a small gain is achieved by moving from $\ring_7$ to $\ring_8$, under \ac{SMP} decoding no performance gain is observed. The reason for this could lay in the \ac{qSC} assumption \eqref{eq:QSCapproxdec} used by the \ac{SMP} decoder,  which holds only in an approximate sense for the case of non-prime rings. The effect is visible over the constant-weight Lee channel too, as depicted in Figure \ref{fig:CWLee}. Both figures show that the \ac{SMP} outperforms the \ac{LSF}, even though in the non-field case for the \ac{SMP} we used the \ac{qSC} assumption \eqref{eq:QSCapproxdec} for the extrinsic channel. We acknowledge that the LSF decoder from \cite{santini2020low} was originally introduced and designed for a special class on \ac{LDPC} codes (namely, for low-Lee-density parity-check codes), and its performance might be enhanced by taking into account the differences between the two code classes. While this point will be subject of further investigations, we believe that an important role in the performance gain under \ac{SMP} decoding relies on its capability to exploit the knowledge of the error marginal distribution.
The block error rate result achieved by \ac{BP} and \ac{SMP} decoding matches well the \ac{DE} analysis, with threshold differences that are reproduced in the finite length results by the gaps among the block error rate curves. As expected, \ac{BP} decoding outperforms \ac{SMP} decoding. Nevertheless, the \ac{SMP} algorithm shows a performance that is appealing for applications demanding low-complexity decoding \cite{santini2020low}.

\begin{figure}[htp]
    \centering\includegraphics[width = \linewidth]{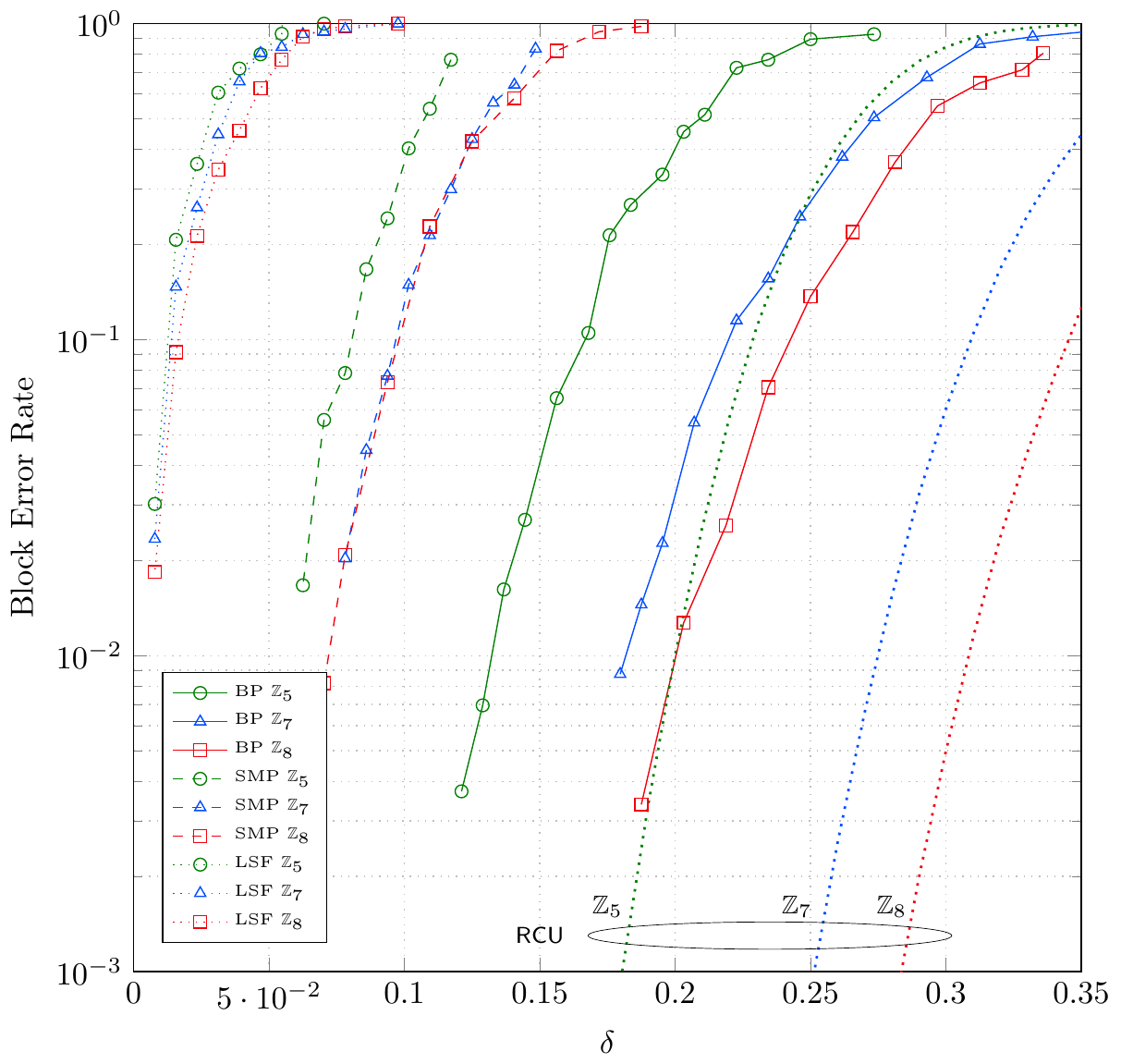}

    \caption{Block error rate vs. $\delta$ for regular $(3, 6)$ nonbinary \ac{LDPC} codes of length $n = 256$. Memoryless Lee channel.}\label{fig:Lee256}
\end{figure}
\begin{figure}[htp]
    \includegraphics[width = \linewidth]{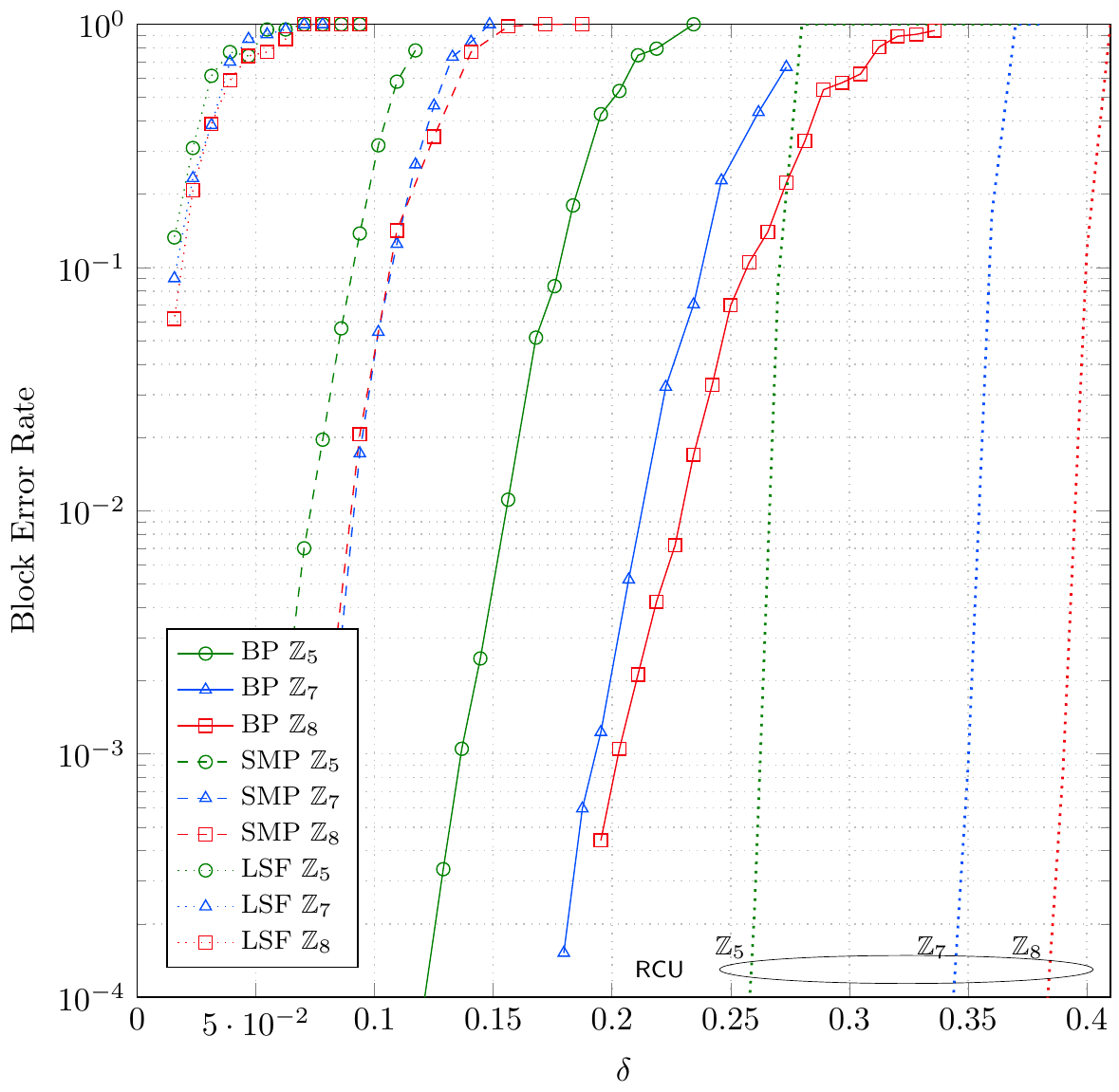}
    \caption{Block error rate vs. $\delta$ for regular $(3, 6)$ nonbinary \ac{LDPC} code ensembles of length $n = 256$. Constant-weight Lee channel.}\label{fig:CWLee}
\end{figure}

\section{Conclusions}\label{sec:conc}

The performance of nonbinary low-density parity-check (LDPC) codes over finite integer rings has been studied, over two channels that arise from the Lee metric. The first channel is a discrete memory-less channel matched to the Lee metric, whereas the second channel adds to each codeword an error vector of constant Lee weight. It is shown that the marginal conditional distribution of the two channels coincides, in the limit of large block lengths. The result is used to provide a suitable marginal distribution to the initialization of the message-passing decoder of LDPC codes. The performance of selected LDPC code ensembles, analyzed by means of density evolution and finite-length simulations under belief propagation (BP) and symbol message passing (SMP) decoding, shows that BP decoding largely outperforms SMP decoding. Nevertheless, the SMP algorithm retains a performance that is appealing for applications (e.g., code-based cryptosystems in the Lee metric) demanding low-complexity decoding.





\bibliography{IEEEabrv,biblio}

\end{document}